\documentclass[conference]{IEEEtran}
\IEEEoverridecommandlockouts
\usepackage{times}
\usepackage{microtype}
\usepackage{graphicx}
\usepackage{subfigure}
\usepackage{booktabs} 
\usepackage{amsmath}
\usepackage{float}
\usepackage{amsthm}
\usepackage{balance}
\usepackage{algorithm} 
\usepackage{algpseudocode}
\usepackage{cite}
\usepackage[keeplastbox]{flushend}
\usepackage{placeins}

\usepackage{hyperref}
\usepackage{amsfonts}
\usepackage{amssymb}
\usepackage{color}
\usepackage{framed}
\newcommand{\MSE}{\mathrm{MSE}}
 
\newcommand{\E}{\mathbb{E}}

\newcommand{\IP}{\mathrm{IP}}
\newcommand{\BCS}{\mathrm{BCS}}
\newcommand{\CBE}{\mathrm{CBE}}
\newcommand{\DOPH}{\mathrm{DOPH}}

\newcommand{\JS}{\mathrm{JS}}
\newcommand{\tb}{\mathrm{\psi}}
\usepackage{bbm}
\usepackage{verbatim}
\newcommand{\IPS}[1]{\langle #1 \rangle}

\newcommand{\F}{\mathrm{F_1}}
\newcommand{\N}{\mathrm{N}}
\newcommand{\Ham}{\mathrm{Ham}}

\newcommand{\dist}{\mathrm{dist}}
\newcommand{\Dist}{\mathrm{D}}
\newcommand{\binsketch}{\mathrm{BinSketch}}
\newcommand{\odd}{\mathrm{OddSketch}}
\newcommand{\simhash}{\mathrm{SimHash}}
\newcommand{\minhash}{\mathrm{MinHash}}

\newcommand{\R}{\mathbb{R}}

\newcommand{\Cos}{\mathrm{Cos}}
 
 \newtheorem{theorem}{Theorem}
 \newtheorem{lemma}[theorem]{Lemma}
 
 \newtheorem{definition}[theorem]{Definition}

\newtheorem{obs}[theorem]{Observation}

\newtheorem{cor}[theorem]{Corollary}

\begin{document}


\title{Efficient Sketching Algorithm for Sparse Binary Data\\
}

\author{\IEEEauthorblockN{1\textsuperscript{st} Rameshwar Pratap}
\IEEEauthorblockA{\textit{School of Computing and Electrical Engineering} \\
\textit{IIT Mandi, H.P.}\\
India \\
rameshwar@iitmandi.ac.in}
\and
\IEEEauthorblockN{2\textsuperscript{nd} Debajyoti Bera}
\IEEEauthorblockA{\textit{Department of Computer Science} \\
\textit{IIIT Delhi}\\
India \\
dbera@iiitd.ac.in}
\and
\IEEEauthorblockN{3\textsuperscript{rd} Karthik Revanuru}
\IEEEauthorblockA{\textit{NALT Analytics} \\
\textit{Bangalore, India}\\
India\\
karthikrvnr@gmail.com}
}

\maketitle

\begin{abstract}
Recent advancement of the WWW, IOT, social network, e-commerce, etc. have generated a large volume of data. These datasets are mostly represented by high dimensional and sparse datasets. Many fundamental subroutines of common data analytic tasks such as clustering, classification, ranking, nearest neighbour search, etc. scale poorly with the dimension of the dataset. In this work, we address this problem and propose a sketching (alternatively, dimensionality reduction) algorithm -- $\binsketch$ (Binary Data  Sketch) -- for sparse binary datasets. $\binsketch$ preserves the binary version of the dataset after sketching and maintains estimates for multiple similarity measures such as Jaccard, Cosine, Inner-Product similarities, and Hamming distance, on the same sketch.  We present a theoretical analysis of our algorithm and complement it with extensive experimentation on several real-world datasets. We compare the performance of our algorithm with the state-of-the-art algorithms on the task of mean-square-error and ranking. Our proposed algorithm offers a  comparable accuracy while suggesting  a significant speedup in the dimensionality reduction time, with respect to the other candidate algorithms. Our proposal is simple, easy to implement, and therefore can be adopted in practice. \footnote{A preliminary version of this paper has been accepted at IEEE-ICDM, 2019.}
\end{abstract}

\section{Introduction}
Due to technological advancements, recent years have witnessed a dramatic increase in our ability to collect data from various sources like WWW, IOT, social media platforms, mobile applications,   finance, and biology.
For example, in many web applications, the volume of datasets are of the terascale order, with trillions of features~\cite{agarwal14a}. 
The high dimensionality incurs high memory requirements and computational cost during the training. 
Further, most of such high dimensional datasets are  sparse, owing to a 
wide adaption of ``Bag-of-words" (BoW) representations. For example: 
in the case of document representation,  word frequency within a document follows power law -- most 
of the words occur rarely in a document,  and higher order shingles occur only once. 
We focus on the binary representation of the datasets which is quite common in several applications~\cite{sibyl,JiangNY07}.

Measuring similarity score of data points under various similarity measures is a fundamental subroutine in several 
applications such as clustering, classification, identifying nearest neighbors, ranking,    and it plays  an important role in various  data mining, machine learning, and information retrieval tasks.
However, due to the 
``curse of dimensionality" a brute-force way of computing the similarity scores in the high dimensional dataset
is infeasible, and at times impossible. In this work, we address this question and propose an efficient dimensionality reduction algorithm for sparse binary datasets that generates a succinct sketch of the dataset while preserving estimates for computing the similarity score between data objects.

 \subsection{Our Contribution}  
We first informally describe our sketching algorithm.

BinSketch: (\textbf{Bin}ary Data \textbf{Sketch}ing) Given a $d$-dimensional binary vector $a\in \{0,1\}^{d}$, our algorithm reduces it to a 
$\N$-dimensional binary vector $a_s\in\{0,1\}^{\N}$, where $\N$ is specified later. 
It randomly maps each bit position (say) $\{i\}_{i=1}^d$  to an integer $\{j\}_{j=1}^{\N}$. 
To compute the $j$-th bit of  $a_s$, 
it checks which bit positions have been mapped to $j$, computes the $\mathtt{bitwise-OR}$ of the bits located at those positions and 
assigns it to $a_s[j].$ 

A simple and exact solution to the problem is to represent each binary vector by a
(sorted) list (or vector) of the indices with value one. In this representation, 
the space required in storing a vector is $O(\tb \log d) $ bits -- as we need  $O(\log d)$ bits for storing each index,
and   there are at most $\tb$ indices with non-zero value (sparsity). Further, the time complexity of computing the (say) inner product of two originally $\tb$-sparse binary vectors is $O (\tb \log d )$.  Therefore,  both the storage as well as   the time complexity of calculating similarity depend on the original dimension $d$ and does not scale for large values of $d$. 
For high dimensional sparse binary data, we show how to construct highly compressed
binary sketches whose length depends only on the data sparsity.
Furthermore, we present techniques to compute similarity between vectors from
their sketches alone. Our main technique is presented in
Algorithm~\ref{alg:ip_est} for inner product similarity and the following
theorem summarizes it.

\begin{theorem}[Estimation of Inner product\label{thm:IP_est}]
Suppose we want to estimate the {\em Inner Product} of $d$-dimensional binary
vectors, whose sparsity is at most $\tb$, with probability at least $1-\rho$. We
can use $\binsketch$ to construct $\N$-dimensional binary sketches where
$\N=\tb \sqrt{\tfrac{\tb}{2} \ln \tfrac{2}{\rho} }$. If $a_s$ and $b_s$
denote the sketches of vectors $a$ and $b$, respectively, then $\IP(a,b)$ can be
    estimated with accuracy $O(\sqrt{\tb\ln \tfrac{6}{\rho}})$ using Algorithm~\ref{alg:ip_est}.
\end{theorem}

We also present Algorithm~\ref{alg:ham_est} for estimating Hamming distance,
Algorithm~\ref{alg:js_est} for estimating Jaccard similarity and
Algorithm~\ref{alg:cos_est} for estimating Cosine similarity; all these
algorithms are designed based on Algorithm~\ref{alg:ip_est} and so follow
similar accuracy guarantees.

\textbf{Extension for categorical data compression.}
Our result can be easily extended for compressing Categorical datasets.   The categorical dataset consists of several categorical features.
Examples of categorical features are sex, weather, days in a 
week,  age group,   educational level, etc.
We  consider  a type of Hamming distance for defining the distance between two categorical data points. For two $d$ dimensional categorical data points $u$ and $v$,  the  distance between them  is defined as follows:
$\Dist(u, v)=\Sigma_{i=1}^d \dist(u[i], v[i])$, where 
\begin{align*}
 \dist(u[i], v[i])=\begin{cases}
    1, & \text{if~} u[i]\neq v[i], \\
    0, & \text{otherwise}.
  \end{cases}
\end{align*}

 In order to use $\binsketch$, we need to preprocess the datasets. 
We  first encode categorical feature via \textit{label-encoding} followed by \textit{one-hot-encoding}.  
In the label encoding step,  features are encoded as integers. For a given feature, if it has $m$ possible values,
we encode them with integers between $0$ and $m-1$. In one-hot-encoding step, we convert the feature value into a  
$m$ length binary string, where  $1$ is located at the position corresponding to the result of the label-encoding step. 
\footnote{Both label-encoder and one-hot-encoder  are available in \texttt{sklearn} as \texttt{labelEncoder} and \texttt{OneHotEncoder} packages.} 
This preprocessing converts categorical dataset to a binary dataset. 
Please note that after preprocessing Hamming distance between the binary version of the data points is equal to the corresponding 
categorical distance $\Dist(,)$, stated above.  
We can now compress the binary version of the 
dataset using $\binsketch$ and due to Algorithm~\ref{alg:ham_est}, the compressed representation maintains the Hamming distance. 

In Section~\ref{section:analysis} we present the proof of
Theorem~\ref{thm:IP_est} where we explain the theoretical reasons behind the effectiveness of
$\binsketch$. As is usually the case for hash functions, practical performance
often outshines theoretical bounds; so we conduct numerous experiments on public
  datasets. Based on our experiment results reported in
Section~\ref{sec:experiment} we make the claim that $\binsketch$ is the best option for
compressing sparse binary vectors while retaining similarity
for many of the commonly used measures. The accuracy obtained is comparable with the state-of-the-art
sketching algorithms, especially at high similarity regions, while taking
almost negligible time compared to similar sketching algorithms proposed so far.



\subsection{Related work} Our proposed algorithm is very similar in nature to the BCS algorithm~\cite{KulkarniP16,JS_BCS}, which 
suggests a randomized bucketing algorithm  where each index of the input is randomly assigned to one of the $O(\tb^2)$ buckets;
$\tb$ denotes the sparsity of the dataset. The sketch of an input vector is obtained by computing
the parity of the bits fallen in each bucket. 
We offer a better compression bound than theirs. 
For a pair of vectors,
their compression bounds are $O(\tb^2)$, while ours is $O(\tb \sqrt{\tb})$. This is also reflected in our empirical evaluations,  on small values of compression length,
$\binsketch$ outperforms $\BCS$.
However, the compression times (or dimensionality reduction time) of both the algorithms are somewhat comparable. 

For Jaccard Similarity, we compare the performance of  our algorithms with $\minhash$~\cite{BroderCFM98}, $\DOPH$~\cite{DOPH} -- a faster variant of $\minhash $, and $\odd$~\cite{oddsketch}. 
We would like to point out some key differences between $\odd$ and $\binsketch$. 
$\odd$  is two-step in nature that takes the sketch obtained by running $\minhash$ on the original data as input,  
and outputs binary sketch which maintains an estimate of the original Jaccard similarity. Due to this two-step nature, 
its compression time is higher (see Table~\ref{tab:comparision} and Figure~\ref{fig:compressiontime}).   The number of $\minhash$ functions 
used in $\odd$ 
(denoted by $k$) is a crucial parameter and the authors suggested using $k$ such that the pairwise symmetric difference is approximately $\N/2$. Empirically they suggest using $k=\N/(4(1-J))$, where $J$ is the similarity threshold.   
We argue that not only tuning $k$ is an important step but it is unclear how this condition will be satisfied for a diverse dataset, 
on the contrary, $\binsketch$ requires no such parameter. Furthermore, $\odd$ doesn't provide any closed form expression to estimate accuracy and confidence. However, the variance of the critical term of their estimator is linear in the size of the sketch, \textit{i.e.} $\N$. Whereas our confidence interval is of the order of $\sqrt{\tb}$ which could be far smaller compared to $\N$, even for non-sparse data. Finally, compared to the Poisson approximation based analysis used in $\odd$, we employed a tighter martingale-based analysis leading to (slightly) better concentration bounds (compare, e.g., the concentration bounds for estimating the size of a set from its sketch).

For Cosine Similarity, we compare $\binsketch$ with $\simhash$~\cite{simhash}, $\CBE$~\cite{CBE} -- a faster variant of $\simhash$, 
 $\minhash$~\cite{ShrivastavaWWW015}, using $\DOPH$~\cite{DOPH}  in the algorithm of~\cite{ShrivastavaWWW015} instead of $\minhash$. 
 For the Inner Product, $\BCS$~\cite{JS_BCS},   Asymmetric MinHash~\cite{ShrivastavaWWW015}, and Asymmetric  $\DOPH$ -- using
 $\DOPH$~\cite{DOPH} in~\cite{ShrivastavaWWW015},  were  the competing algorithms.      
In all these similarity measures,  for sparse binary datasets, our proposed algorithm is faster,    while simultaneously offering almost a similar performance as compared to the baselines. We experimentally compare the performance on several
real-world datasets and observed the results that are in line with these observations. 
Further, in order to get a sketch of size $\N$, our algorithm requires a lesser number of random bits, and requires only one pass to the datasets.  
These are the major reasons due to which we obtained good speedup in compression time.  We summarize this comparison in Table~\ref{tab:comparision}.
Finally, a major advantage of our algorithm, similar to~\cite{KulkarniP16,JS_BCS}, is that it gives one-shot sketching by maintaining estimates of multiple similarity measures in the same sketch; this is in contrast to usual sketches that are customized for a specific similarity.
 \begin{center}
   \begin{table}     
     \small
       \caption{\footnotesize{A comparison among the candidate algorithms, on the number of random bits and the compression time,  to get a sketch of length $\N$ of a single data object. Compression time includes both (i) time required to generate hash function, which is of order the number of random bits, (ii) time required to generate the sketch using the hash functions. The parameter $k$ for $\odd$ denotes the number of permutations required by an intermediate $\minhash$ step.
}}\label{tab:comparision}
  \begin{tabular}{|c|c|l|}
   \hline
    Algorithm                  				&No of random bits                 & Compression time    \\       
  \hline
    $\binsketch$              			        & $O( d \log \N) $                    &$O(d \log \N+\tb)$\\ 
    $\BCS$~\cite{KulkarniP16,JS_BCS}                    & $O( d \log \N) $                    &$O(d \log \N+\tb)$\\ 
    $\DOPH$~\cite{DOPH}                                 &$O(d \log d)$                        &$O(d \log d+\tb + \N)$\\
    $\CBE$~\cite{CBE}                                   &$O(d)$                               &$O( d \log d)$\\
    $\odd$~\cite{oddsketch}                             &$O(k(d \log d+\N))$                 & $O(k(d \log d+\N+\tb))$\\
    $\simhash$~\cite{simhash}                           &$O(d \N)$                            &$O((d +\tb)\N)$\\
    $\minhash$~\cite{BroderCFM98}                       &$O((d \log d) \N)$                   &$O((d \log d +\tb)\N)$\\
    \hline
%
\end{tabular}
\end{table}
\end{center}
\vspace{-0.5cm}
\paragraph{Connection with Bloom Filter} $\binsketch$ appears structurally similar to a Bloom filter with one hash function. 
The standard Bloom filter is a space-efficient data-structure for \textit{set-membership} queries; however, there is an alternative 
approach that can be used to estimate the intersection between two sets~\cite{BroderM03}. However, it is unclear how estimates for other 
similarity measures can be obtained. We answer this question positively and suggest estimates for all the four similarity measures in the same sketch. We also show that our estimates are strongly concentrated around their expected values. 

\subsection{Applicability of our results}\label{subsec:application}
For high dimensional sparse binary datasets,   $\binsketch$ due to its simplicity, efficiency, and performance, can
be used in numerous applications which require a sketch preserving Jaccard, cosine, Hamming distance or inner product similarity.
%

\paragraph*{Scalable Ranking and deduplication of documents}
Given a corpus of documents and a set of query documents, a goal is to find all
documents in the corpus that are ``similar'' to query documents under a given
similarity measure (e.g., Jaccard, cosine, inner product). 
This problem is a fundamental sub-routine in many applications like
near-duplicate data detection~\cite{MankuJS07,Henzinger06,BroderCPM00}, efficient document
similarity search~\cite{JiangS11,ShrivastavaWWW015}, plagiarism
detection~\cite{BuyrukbilenB13,BroderCPM00}, etc. and dimensionality reduction
is one way to address this problem.
In Subsection~\ref{subsection:ranking} we provide empirical validation that $\binsketch$ offers significant speed-up in dimensionality reduction while offering a comparable accuracy.

 \paragraph*{Scalable Clustering of documents}  
$\binsketch$ can be used in scaling up the performance of several clustering algorithms, in the case of high-dimensional and sparse datasets. 
For instance, in the case of  Spherical $k$-means clustering,  which is the problem of  clustering data points  using Cosine Similarity, one can use~\cite{SPKM}; 
and for   $k$-mode clustering, which is clustering using Hamming Distance, one can use $k$-mode~\cite{kmode}, on the sketch obtained by $\binsketch$.
 \paragraph*{Other Applications}
Beyond the above-noted applications, sketching techniques have been used widely in application such as 
 Spam detection~\cite{broder1997resemblance}, 
compressing social networks~\cite{ChierichettiKLMPR09} 
all 
 pair similarity~\cite{BayardoMS07}, Frequent Itemset Mining~\cite{ChakaravarthyPS09}. 
As $\binsketch$   offers significant speed-up in dimensionality reduction time and simultaneously provides a succinct and accurate sketch, it helps in scaling up the performance of the respective algorithms.

 \section{Background}\label{sec:background}
 \begin{center}
 \scalebox{0.85}{
\begin{tabular}{|c|l|}
\hline
 \multicolumn{2}{|c|}{\bf Notations}\\
 \hline
 $\N$ & dimension of  the compressed data.  \\
 \hline
$\tb$ & sparsity bound. \\ 
\hline
$u[i]$ & $i$-th bit position  of binary  vector $u.$\\
\hline
$|u|$ & number of $1$'s in the binary vector $u$.\\
\hline
$\Cos(u, v)$ & Cosine similarity between    $u$ and $v.$\\
\hline
$\JS(u, v)$ & Jaccard similarity between   $u$ and $v.$\\
\hline
 $\Ham(u, v)$& Hamming distance between  $u$ and $v.$\\
 \hline
$\IP(u, v)$ & Inner product between   $u$ and $v.$\\
\hline
 \end{tabular}
 }
\end{center}
 \paragraph{SimHash  for Cosine similarity~\cite{simhash,GoemansW95}.} 
 The Cosine similarity between a
 pair of vectors $u,v\in \mathbb{R}^d$ is defined as $\langle u, v \rangle /
\|u\|_2 \cdot \|v\|_2$.
To compute a sketch of  a  vector $u$, $\simhash$~\cite{simhash} generates a random vector $ {r}\in \{-1, +1\}^d$,  
with each component chosen uniformly at random from $\{-1, +1\}$ and a 1-bit sketch is computed as
\begin{align*}
 \simhash^{(r)}(u)=\begin{cases}
    1, & \text{if $\langle u,  {r}\rangle $} \geq 0.\\
    0, & \text{otherwise}.
  \end{cases}
\end{align*}

$\simhash$ was shown to preserve inner product in the following manner~\cite{GoemansW95}.
Let $\theta$ be an angle such that $\cos \theta = \IPS{u,v}/\|u\| \cdot \|v\|$.
Then,
\[
 \Pr[\simhash^{(r)}(u) = \simhash^{(r)}(v)]=1-\frac{\theta}{\pi}, 
\]


 \paragraph{MinHash for Jaccard and Cosine similarity.}
The Jaccard similarity between a
 pair of set $u, v\subseteq \{1, 2, \ldots d\}$
is defined as  
$\JS(u, v)=\frac{|u \cap v|}{|u \cup v|}.$
Broder et al.~\cite{BroderCFM98} suggested an algorithm -- $\minhash$ --
  to compress a collection of sets while preserving the  Jaccard similarity between any pair of sets.
 Their technique includes taking a random permutation 
of $\{1, 2, \ldots, d\}$ and assigning a value to each set which maps to  
minimum  under that permutation. 
\begin{definition}[Minhash~\cite{BroderCFM98}]\label{defn:minwise}
    Let $\pi$ be a random permutation over $\{1, \ldots, d\}$, then for a set $u\subseteq \{1,\ldots d\}$
    $h_\pi(u) = \arg\min_i \pi(i)$ for $i \in u$.
\end{definition}
   
It was then shown by Broder et al.~\cite{BroderCFM98,BroderCPM00} that
    \begin{align*}\label{eq:minhash1}
 \Pr[h_\pi(u)=h_\pi(v)]=\frac{|u\cap v|}{|u \cup v|}.
\end{align*}

Exploiting a similarity between Jaccard similarity of sets and Cosine similarity
of binary vectors, it was shown how to use $\minhash$ for constructing sketches
for Cosine similarity in the case of sparse binary data~\cite{Shrivastava014}.

\paragraph{BCS for sparse binary data~\cite{JS_BCS,KulkarniP16}.} 
For sparse binary dataset, $\BCS$ offers a sketching algorithm  that simultaneously 
preserves Jaccard similarity, Hamming distance and inner product. 

 \begin{definition}[BCS]\label{defi:bcs}
     Let $\N$ be the number of buckets. Choose a random mapping $b$ from $\{1
     \ldots d\}$ to $\{1, \ldots N\}$. Then a vector $u \in \{0, 1\}^d$ is
 compressed to a vector $u_s \in \{0, 1\}^{\N}$ as follows:
 \[u_s[j] = \sum_{i : b(i) = j} u[i]  \pmod 2.\]  
\end{definition}

\section{Analysis of $\binsketch$}\label{section:analysis}

Let $a$ and $b$ denote two binary vectors in $d$-dimension, and $|a|$, $|b|$ denotes the number of $1$ in $a$ and $b$. 
Let $a_s, b_s\in \{0, 1\}^{\N}$  denote the compressed representation of $a$ and $b$, where  $\N$ denotes the  compression length (or reduced dimension).
In this section we will explain our sketching method $\binsketch$ and give
theoretical bounds on its efficacy.

\begin{definition}[$\binsketch$]
    Let $\pi$ be a random mapping from $\{1, \ldots d\}$ to $\{1, \ldots {\N}\}$.
    Then a vector $a \in \{0,1\}^d$ is compressed into a vector $a_s \in
    \{0,1\}^{\N}$ as $$a_s[j] = \bigvee_{i : \pi(i)=j} a[i]$$
\end{definition}

Constructing a $\binsketch$ for a dataset involves first, generating a random
mapping $\pi$, and second, hashing each vector in the dataset using $\pi$.
There could be ${\N}^d$ possible mappings, so choosing $\pi$ requires
$O(\log({\N}^d)) = O(d \log N)$ time and that many random bits. Hashing a vector $a$ involves
only looking at the non-zero bits in $a$ and that step takes time $O(\tb)$ since
$|a| \le \tb$. Both these costs compete favorably with the existing algorithms as
tabulated in Table~\ref{tab:comparision}.

\subsection{Inner-product similarity}

The sketches, $a_s$'s do not quite ``preserve'' inner-product by themselves,
but are related to the latter in the following sense. We will use $n$ to denote
$1-\frac{1}{\N} \in (0,1)$; it will be helpful to note that $n \to 1$ as $\N$
increases.

\begin{lemma}\label{lemma:E_ip_ab}
\begin{align*}
    & 1.~\E(|a_s|/\N) = (1-n^{|a|})\\
    & 2.~\E(\IPS{a_s, b_s}/\N) = \\
    & (1-n^{|a|})(1-n^{|b|}) + n^{|a|+|b|}
    \left[\left( \frac{1}{n}\right)^{\IPS{a,b}} - 1 \right] = \\
    & 1 - n^{|a|} - n^{|b|} + n^{|a|+|b|+\IPS{a,b}}
\end{align*}
\end{lemma}

\begin{proof}
    It will be easier to identify $a \in \{0,1\}^d$ as a subset of $\{1, \ldots
    d\}$. 
    The $j$-th bit of $a_s$ can be set only by some element in $a$ which can
    happen with probability $(1-(1-\tfrac{1}{\N})^{|a|})$. The $j$-th bit of both $a_s$ and $b_s$ is set if it is set by some element in
$a\cap b$, or if it is set simultaneously by some element in $a \setminus (a\cap
    b) = a \setminus b$ and by
another element in $b \setminus (a\cap b)$. This translates to the following probability
    that some particular bit is set in both $a_s$ and $b_s$.
\begin{align*}
    & \left(1-n^{|a\cap b|}\right) +
     n^{|a\cap b|} \left(1- n^{|a\setminus b|}   \right)\left( 1-n^{|b \setminus
    a|}   \right)\\
    & = 1 - n^{|a|} - n^{|b|} + n^{|a|+|b|-|a \cap b|}\\
    & = (1-n^{|a|})(1-n^{|b|}) + n^{|a|+|b|}\left( \frac{1}{n^{|a \cap b|}} - 1
    \right)
\end{align*}
The lemma follows from the above probabilities using the linearity of expectation.
\end{proof}

Note that the above lemma
allows us to express $\IPS{a,b}$ as
$$\IPS{a,b} = |a| + |b| - \tfrac{1}{\ln n} \ln
\left(n^{|a|}+n^{|b|}+\frac{\E(\IPS{a_s,b_s})}{\N}-1 \right)$$
Algorithm~\ref{alg:ip_est} now explains how to use this result to approximately calculate
$\IPS{a,b}$ using their sketches $a_s$ and $b_s$.

    \begin{algorithm}
	\hspace*{\algorithmicindent} \textbf{Input:} Sketches $a_s$ of $a$ and $b_s$ of
	$b$
	\begin{algorithmic}[1]
	    \State Estimate $\E[|a_s|]$ as $n_{a_s}=|a_s|$, $\E[|b_s|]$ as $n_{b_s}=|b_s|$
	    \State Estimate $\E[\IPS{a_s, b_s}]$ as $n_{a_s,b_s}=\IPS{a_s, b_s}$
	    \State Approximate $|a|$ as $n_a=\ln(1-\frac{n_{a_s}}{\N})/\ln(n)$ 
	    and \mbox{$|b|$ as $n_b =\ln(1-\frac{n_{b_s}}{\N})/\ln(n)$}
	    \State \Return approximation of $\IPS{a,b}$ as $$n_{a,b} = n_a + n_b - \tfrac{1}{\ln n} \ln
		\left(n^{n_a}+n^{n_b}+\frac{n_{a_s,b_s}}{\N}-1 \right)$$ 
	\end{algorithmic}
	\caption{$\binsketch$ estimation of $\IP(a,b)$ \label{alg:ip_est}}
    \end{algorithm}

We will prove that Algorithm~\ref{alg:ip_est} estimates $\IPS{a,b}$ with high
accuracy and confidence if we use $\N=\tb \sqrt{\frac{\tb}{2} \ln
\frac{2}{\delta}}$; $\delta$ can be set to any desired probability
of error and we assume that the sparsity $\psi$ is not too small, say at least
20. Our first result proves that the $n_{a_s}$ estimated above is a
good approximation of $\E[|a_s|]$; exactly identical result holds for $b_s$ and
$n_{b_s}$ too.

\begin{lemma}\label{lemma:1} With probability at least $1-\delta$, it holds that $$\Big|
    n_{a_s} - \E[|a_s|] \Big| < \sqrt{\frac{\tb}{2}\ln \frac{2}{\delta}}$$
\end{lemma}

\begin{proof}
The proof of this lemma is a simple adaptation of the computation of the
expected number of non-empty bins in a balls-and-bins experiment that is found
in textbooks and done using Doob's martingale.
Identify the random mapping $\pi(a)$, where the number of 1's in $a$ is
    denoted by $|a|$, as throwing $|a|$ black balls (and $d-|a|$ ``no''-balls), one-by-one, into $\N$ bins chosen
    uniformly at random. Supposing we only consider the black balls in the bins, then $a_s[j]$ is an indicator variable for the event
    that the $j$-th bin is non-empty and the number of non-empty bins
    can be shown to be concentrated around their expectation~\footnote{Using $F$ to denote the number of non-empty bins and $m$ the
number of balls, Azuma-Hoeffding inequality states that $\Pr\Big[|F - \E[F]| \ge \lambda\Big] \le 2
\exp(-2\lambda^2/m)$ (see Probability and Computing, Mitzenmacher and Upfal,
Cambridge Univ. Press).}. Since the number of non-empty bins correspond to
    $|a_s|$, this concentration bound can be directly applied for proving the
    lemma.

Let $\mathcal{E}$ denote the event in the statement of the lemma. Then,
\[
    \Pr[\bar{\mathcal{E}}] \le \Pr\left[ \Big| |a_s| -
    \E\big[|a_s|\big] \Big| \ge \sqrt{\frac{|a|}{2}\ln
    \frac{2}{\delta}}\right] \le \delta 
\]
where $|a| \le \tb$ is used for the first inequality and the stated bound, with
    $m=|a|$, is used for the second inequality.
\end{proof}

Similar, but more involved, approach can be used to prove that $n_{a_s,b_s}=\IPS{a_s,b_s}$ is a
good estimation of $\E[\IPS{a_s,b_s}]$. 

\begin{lemma}\label{lemma:3}
    With probability at least $1-\delta$, it holds that
    $$\Big| n_{a_s,b_s} -
\E[\IPS{a_s,b_s}] \Big| < \sqrt{\frac{\tb}{2}\ln \frac{2}{\delta}}$$
\end{lemma}

\begin{proof}
For a given $a,b \in \{0,1\}^d$, lets
partition $\{1,\ldots d\}$ into parts $C$ (consisting of positions at which both
$a$ and $b$ are 1), $D$ (positions at which $a$ is 1 and $b$ is 0), $E$
(positions at which $a$ is 0 and $b$ is 1) and $F$ (the rest). Any random
mapping $\pi$ can treated as throwing $|C|$ grey balls, $|D|$ white balls,
$|E|$ black balls, and $d-|C|-|D|-|E|$ ``no''-balls randomly into $\N$ bins.
Suppose we say that a bin is ``greyish'' if it either contains some grey
ball or both a white and a black ball. The number of common 1-bits in $a_s$ and
$b_s$ (that is $n_{a_s,b_s}=\IPS{a_s,b_s}$) is now equal to the number of greyish
bins. Observe that when any ball lands in some bin, say $j$, the number of
greyish bins either remains same or increases by 1; therefore, we can say that
the count of the greyish bins satisfies Lipschitz condition. This allows us to
apply Azuma-Hoeffding inequality as above and prove the lemma; we will also need the fact that the number of greyish bins is at most $\tb$. 
\end{proof}

The next lemma allows us to claim that our estimation of $|a|$ is also within
reasonable bounds. It should be noted that our sketches $|a_s|$ do not
explicitly save the number of 1's in $a$, so it is necessary to compute this number from our sketches; furthermore, since this estimate is not used elsewhere, we
do not mandate it to be an integer either.

\begin{lemma}\label{lemma:2}
    With probability at least $1-\delta$, it holds that
    $$\Big| |a| - n_a \Big| < \frac{4}{\tb \ln \tfrac{1}{n}} =
    4\sqrt{\frac{\tb}{2} \ln \frac{2}{\delta}}$$
\end{lemma}

\begin{proof}Based on Lemma~\ref{lemma:E_ip_ab} and Algorithm~\ref{alg:ip_est},
    $n^{|a|}-n^{n_a} = [n_{a_s} - \E(|a_s|)]/N$. For the proof we use the upper bound
    given in Lemma~\ref{lemma:1} that holds with probability at least
    $1-\delta$. We need a few results before proceeding that are based on the standard inequality
    $\ln(1-x) \le -x$ for $0 < x < 1$.

    \begin{obs}
	$\ln \frac{1}{n} \ge \frac{1}{N}$ ($\because$ $\ln n = \ln(1-1/N) \le
	-\frac{1}{N}$)
    \end{obs}

    \begin{obs}\label{obs:2}$n_a=\ln (1-\frac{n_{a_s}}{N})/\ln n \le
	\frac{n_{a_s}}{N}/\ln(\frac{1}{n})$. Since $n_{a_s} \le N$, we get that
	$n_a \le N$.
    \end{obs}

    \begin{obs}\label{obs:1}$n^{n_a} \ge \frac{1}{2}$.
    \end{obs}
    
A proof of the above observation follows using simple algebra and the result of Lemma~\ref{lemma:1}.  We defer it to the full version of the paper. 
    We use these observations for proving two possible cases of the
    lemma. We will use the notation $\Delta = \Big|n_a - |a|\Big|$.

    \noindent{\bf case (i) $|a| \le n_a$:} In this case $\Delta=n_a - |a|$ and $$n^{|a|}-n^{n_a} = [n_{a_s} - \E(|a_s|)]/N$$
    For the R.H.S., $[n_{a_s} - \E(|a_s|)]/N \le 1/\psi$ by
    Lemma~\ref{lemma:1}.\\
    For the L.H.S., we can write $n^{|a|} - n^{n_a} =
    n^{|a|}(1-n^{n_a - |a|}) \ge n^{\psi}(1-n^{\Delta})$ as $|a| \le
    \psi$. Furthermore, $n^{\psi} = (1-\frac{1}{N})^{\psi} \ge
    1 - \frac{\psi}{N} > \frac{1}{2}$ since $\frac{\psi}{N} =
    1/\sqrt{\frac{\psi}{2}\ln\frac{2}{\delta}} < \frac{1}{2}$ for reasonable values
    of $\psi$ and $\delta$.\\
    Combining the bounds above we get the inequality 
    \mbox{$\frac{1}{2}(1-n^{\Delta}) < 1/\psi$} that we will further process below.
    
    \noindent{\bf case (ii) $n_a \le |a|$:} In this case $\Delta=|a|-n_a$ and $$n^{n_a} - n^{|a|} = [\E(|a_s|)
    - n_{a_s}]/N$$
    As above, R.H.S. is at most $1/\psi$ using Lemma~\ref{lemma:1} and L.H.S. can be written as $n^{n_a}(1 -
    n^{\Delta})$. Further using Observation~\ref{obs:1} we get the inequality,
    $\frac{1}{2}(1-n^{\Delta}) \le 1/\psi$.

    For both the above cases we obtained that $\frac{1}{2}(1-n^{\Delta}) \le
    1/\psi$, i.e., $1-n^{\Delta} \le 2/\psi$. This gives us
    that $\Delta\ln n \ge \ln(1-2/\psi) \ge \frac{-2/\psi}{1-2/\psi} =
    \frac{-2}{\psi-2}$ employing the
    known inequality $\ln(1+x) \ge \frac{x}{x+1}$ for any $x > -1$. Since $n \in
    (0,1)$, we get the desired upper bound $\Delta \le \frac{2}{\psi-2}
    \frac{1}{\ln \frac{1}{n}} \le \frac{4}{\psi \ln \frac{1}{n}}$ (since
    $\frac{\psi}{2} \le \psi-2$ for $\psi \ge 4$) $\le 4\sqrt{\frac{\tb}{2} \ln
    \frac{2}{\delta}}$ (using Observation~\ref{obs:1}).
\end{proof}

Of course a similar result holds for $|b|$ and $n_b$ as well.
The next lemma similarly establishes the accuracy of our estimation of
$\IPS{a,b}$.

\begin{lemma}\label{lemma:IP_estimation}
    With probability at least $1-3\delta$, it holds that
    $$\Big|\IPS{a,b} - n_{a,b} \Big| < 14\sqrt{\frac{\tb}{2}\ln\frac{2}{\delta}}$$
\end{lemma}

We get the following from Algorithm~\ref{alg:ip_est} and
Lemma~\ref{lemma:E_ip_ab}.
\begin{align*}
    \IPS{a,b} & = |a|+|b|+ \tfrac{1}{\ln \tfrac{1}{n}} \ln\left[n^{|a|} + n^{|b|} +
    \frac{\E[\IPS{a_s,b_s}]}{\N}-1\right]\\
    n_{a,b} & = n_a + n_b + \tfrac{1}{\ln \tfrac{1}{n}}\ln \left( n^{n_a} +
    n^{n_b} + \frac{n_{a_s,b_s}}{\N} - 1 \right)
\end{align*}
in which $|a| \approx n_a$ (Lemma~\ref{lemma:2}), $|b| \approx n_b$ (similarly),
and $\E[\IPS{a_s,b_s}] \approx n_{a_s,b_s}$ (Lemma~\ref{lemma:3}), each
happening with
probability at least $1-\delta$.
The complete proof that $n_{a,b}$ is a good approximation of $\IPS{a,b}$ is mostly
algebraic analysis of the above facts and we defer it  the full version of the paper. 

Theorem~\ref{thm:IP_est} is a direct consequence of
Lemma~\ref{lemma:IP_estimation} for reasonably large $\tb$ (say, beyond $20$) and small
$\delta$ (say, less than $0.1$).

\subsection{Hamming distance} The Hamming distance and the inner product
similarity of two binary vectors $a$ and $b$ are related as
\[
 \Ham(a, b)=|a|+|b|-\IP(a, b) \label{eq:Hamming}
\]

The technique used in the earlier subsection can be used to estimate the Hamming
distance in a similar manner.

    \begin{algorithm}
	\hspace*{\algorithmicindent} \textbf{Input:} Sketches $a_s$ of $a$ and $b_s$ of
	$b$
	\begin{algorithmic}[1]
	    \State Calculate $n_{a}$, $n_{b}$, $n_{a,b}$ as done in Algorithm~\ref{alg:ip_est}
	    \State \Return approx. of $\Ham(a,b)$ as 
		$ ham_{a,b} = n_a + n_b - n_{a,b}$
	\end{algorithmic}
	\caption{$\binsketch$ estimation of $\Ham(a,b)$ \label{alg:ham_est}}
    \end{algorithm}

\subsection{Jaccard similarity} 
The Jaccard similarity between a pair of binary vectors $a$ and $b$ can be
computed from their Hamming distance and their inner product.
\[
 \JS(a, b)=\frac{\IP(a, b)}{\Ham(a, b)+\IP(a, b)}
\]

This paves way for an algorithm to compute Jaccard similarity from $\binsketch$.
    \begin{algorithm}
	\hspace*{\algorithmicindent} \textbf{Input:} Sketches $a_s$ of $a$ and $b_s$ of
	$b$
	\begin{algorithmic}[1]
	    \State Calculate $n_{a,b}$ using Algorithm~\ref{alg:ip_est}
	    \State Calculate $ham_{a,b}$ using Algorithm~\ref{alg:ham_est}
	    \State \Return approx. of $\JS(a,b)$ as 
		$ JS_{a,b} = \dfrac{n_{a,b}}{n_{a,b} + ham_{a,b}}$
	\end{algorithmic}
	\caption{$\binsketch$ estimation of $\JS(a,b)$ \label{alg:js_est}}
    \end{algorithm}

\subsection{Cosine similarity} The cosine similarity between a pair binary
vectors $a$ and $b$ is defined as: 
 \[
     \Cos(a, b)=\IP(a, b)\Big/\sqrt{|a|\cdot|b|}
 \]

An algorithm for estimating cosine similarity from binary sketches is straight
forward to design at this point.
    \begin{algorithm}
	\hspace*{\algorithmicindent} \textbf{Input:} Sketches $a_s$ of $a$ and $b_s$ of
	$b$
	\begin{algorithmic}[1]
	    \State Calculate $n_a$, $n_b$, $n_{a,b}$ as done in Algorithm~\ref{alg:ip_est}
	    \State \Return approx. of $\Cos(a,b)$ as 
		$ cos_{a,b} = n_{a,b}\Big/\sqrt{n_a \cdot n_b}$
	\end{algorithmic}
	\caption{$\binsketch$ estimation of $\Cos(a,b)$ \label{alg:cos_est}}
    \end{algorithm}

It should be possible to prove that Algorithms \ref{alg:ham_est},
\ref{alg:js_est} and \ref{alg:cos_est} are accurate and low-error estimations of
Hamming distance, Jaccard similarity and cosine similarity, respectively;
however, those analysis are left out of this paper.

 \section{Experiments}\label{sec:experiment}
  
\begin{figure}[!ht]
\centering
\includegraphics[height=15cm,width=9cm]{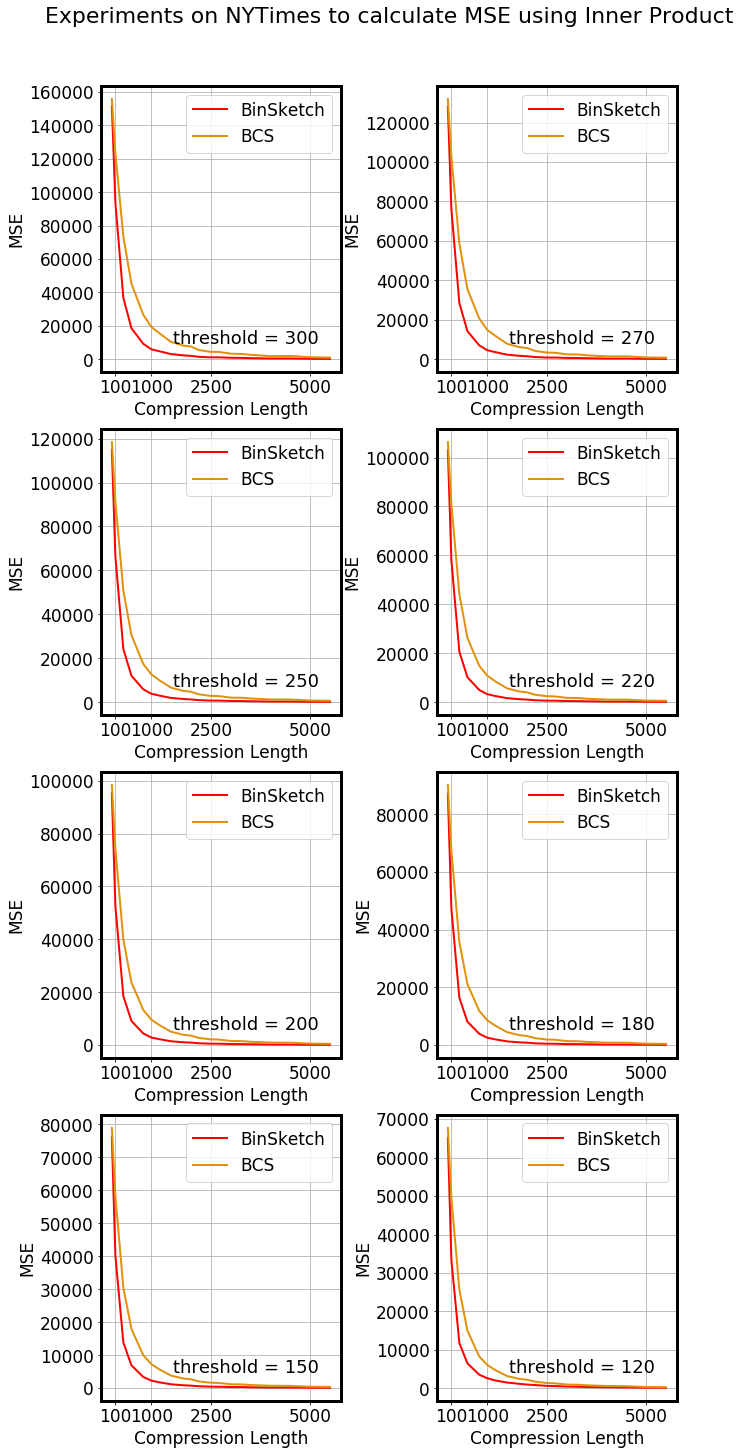}
\caption{\footnotesize{Comparison of $\MSE$ measure on  NYTimes datasets. A lower value is an indication of better performance.}}
\label{fig:MSE_IP_NYTimes}
\end{figure}

\paragraph{Hardware description} We performed our experiments on a machine having the following configuration: 
CPU: Intel(R) Core(TM) i5-3320M CPU @ 2.60GHz x 4; 
Memory: 7.5 GB;
OS: Ubuntu 18.04; 
Model: Lenovo Thinkpad T430.

To reduce the effect of
randomness, we repeated each experiment several times and took
the average.
Our implementations did not employ any special optimization.
   \paragraph*{Datasets}
 The experiments were performed on publicly available datasets - namely,  
 NYTimes news articles (number of points = $300000$, dimension =  $102660$),
 Enron Emails (number of points = $39861$, dimension =  $28102$), and   
 KOS blog entries (number of points = $3430$, dimension =  $6906$) from the UCI machine learning 
 repository~\cite{UCI}; and BBC News Datasets (number of points = $2225$, dimension = $9635$)~\cite{bbc}. 
 We considered the entire corpus of KOS and BBC News datasets, while for
NYTimes, ENRON datasets we sampled $5000$ data points.
\paragraph{Competing Algorithms} For our experiments we have used three   similarity measures: 
Jaccard Similarity,  Cosine Similarity, and Inner Product.  
 For the Jaccard Similarity, $\minhash$~\cite{BroderCFM98}, Densified One Permutation Hashing ($\DOPH$) -- a faster variant of $\minhash$ --~\cite{DOPH},  $\BCS$~\cite{JS_BCS}, and $\odd$~\cite{oddsketch}  were  the competing algorithms. 
   $\odd$  is  two-step in nature,  which takes the sketch obtained by running $\minhash$ on the original data as input,  
and outputs binary sketch which maintains an estimate of the original Jaccard similarity.  As suggested by authors, we use the number of $\minhash$ permutations  $k = \N/(4(1-J))$, where $J$ is the similarity threshold.
{We upper bound $k$ with $5500$ which is the maximum number of permutations used by $\minhash$.}   
  For the Cosine Similarity, $\simhash$~\cite{simhash},  Circulant Binary Embedding ($\CBE$) -- a faster variant of $\simhash$ --~\cite{CBE}, 
     $\minhash$~\cite{Shrivastava014},  $\DOPH$~\cite{DOPH} in the algorithm of~\cite{Shrivastava014} instead of $\minhash$,  were  the competing algorithms.      
  For the Inner Product, $\BCS$~\cite{JS_BCS},   Asymmetric MinHash~\cite{ShrivastavaWWW015}, and Asymmetric  $\DOPH$ ($\DOPH$~\cite{DOPH} in the algorithm of~\cite{ShrivastavaWWW015}),  were  the competing algorithms.

\subsection{Experiment $1$: Accuracy of Estimation}\label{subsection:neg_log} 
In this task, we evaluate the \textit{fidelity of the estimate}  of $\binsketch$ on various similarity regimes. 

\paragraph{Evaluation Metric}  
 To understand the behavior of $\binsketch$ on various similarity regimes, we extract similar pairs -- pair of data objects whose similarity is higher than certain threshold --from the datasets. We used Cosine,  Jaccard, and Inner Product as our measures.  
For example: for Jaccard/Cosine case for the threshold value $0.95$,  we considered only those pairs whose    similarities are higher than $0.95$. We used mean square error $(\MSE)$ as our evaluation criteria.   Using $\binsketch$ and other candidate algorithms, we compressed the datasets to various values of compression length $\N$. 
We then calculated the $\MSE$ for all the algorithms, for different values of $\N$.   For example, in order to calculate the  $\MSE$ of $\binsketch$ with respect to the ground truth result, for every pair of data points, we calculated the square of the difference between their estimated similarities after the result of $\binsketch$,  and the corresponding ground truth similarity.  
 We added these values for all such pairs and calculated its mean. For Inner Product, we used this absolute value, and for Jaccard/Cosine similarity 
we computed its negative logarithm base $e$.  A smaller $\MSE$ corresponds to a larger   $-\log(\MSE)$, therefore, a higher value $-\log(\MSE)$ is an indication of better performance.

\paragraph{Insights}
 We summarize our results in Figures~\ref{fig:neglogMSE}, and \ref{fig:MSE_IP_NYTimes} for Cosine/Jaccard Similarity and Inner Product, respectively. 
For Cosine Similarity, $\binsketch$  consistently remains to be better than the other candidates. While for 
  Jaccard Similarity, it significantly outperformed \textit{w.r.t.} $\BCS$,  $\DOPH$ and $\odd$, while its performance was comparable \textit{w.r.t.} $\minhash$.
Moreover, for Inner product \ref{fig:MSE_IP_NYTimes} results, $\binsketch$ significantly outperformed \textit{w.r.t.} $\BCS$. 
\footnotetext{We observed a similar pattern for both $\MSE$ as well as Ranking experiments on other datasets/similarity measures as well. We defer those plot to the full version of the paper. }

 \subsection{Experiment $2$: Ranking}\label{subsection:ranking}
  \paragraph*{Evaluation Metric}
 In this experiment, given a dataset and a set of query points, the aim is to find all the points that are similar to the query points, 
 under the given similarity measure. To do so, we randomly, partition the dataset into two parts --  $90\%$ and $10\%$. The bigger partition
 is called as the \textit{training partition}, while the smaller one is called as \textit{querying partition}. We call each vector of the 
 querying partition as a query vector.  For each query vector, we compute the points in the training partition whose   similarities
 are higher than a certain threshold. 
For Cosine and Jaccard Similarity, we used the threshold values from the set $\{ 0.95, 0.9, 0.85, 0.8, 0.6, 0.5, 0.2, 0.1\}$.  
For Inner Product,    we first found out the maximum existing Inner product in the dataset, and then set the thresholds accordingly.
 For every query point, we first find all the similar points in the uncompressed dataset, which we call as ground truth result. We then compress the dataset, using the candidate algorithms, on various values of compression lengths. 
To evaluate the performance of the competing algorithms, we used the \textit{accuracy-precision-recall-$\F$ score}   as our standard measure. 
If the set $\mathcal{O}$ denotes the ground truth result (result on the uncompressed dataset), and the 
 set $\mathcal{O'}$ denotes the results on the compressed datasets, then 
 accuracy = ${|\mathcal{O}\cap \mathcal{O'}|}/{ |\mathcal{O}\cup \mathcal{O'}|}$, 
 precision = ${|\mathcal{O}\cap \mathcal{O'}|}/{|\mathcal{O'}|}$,  
 recall = ${|\mathcal{O}\cap \mathcal{O'}|}/{|\mathcal{O}|}$, and 
 $\F~\text{score}=(2\cdot \text{precision}\cdot \text{recall})/(\text{precision}+\text{recall}).$

\begin{figure*}
\centering
\includegraphics[height=7.4cm,width=16.2cm]{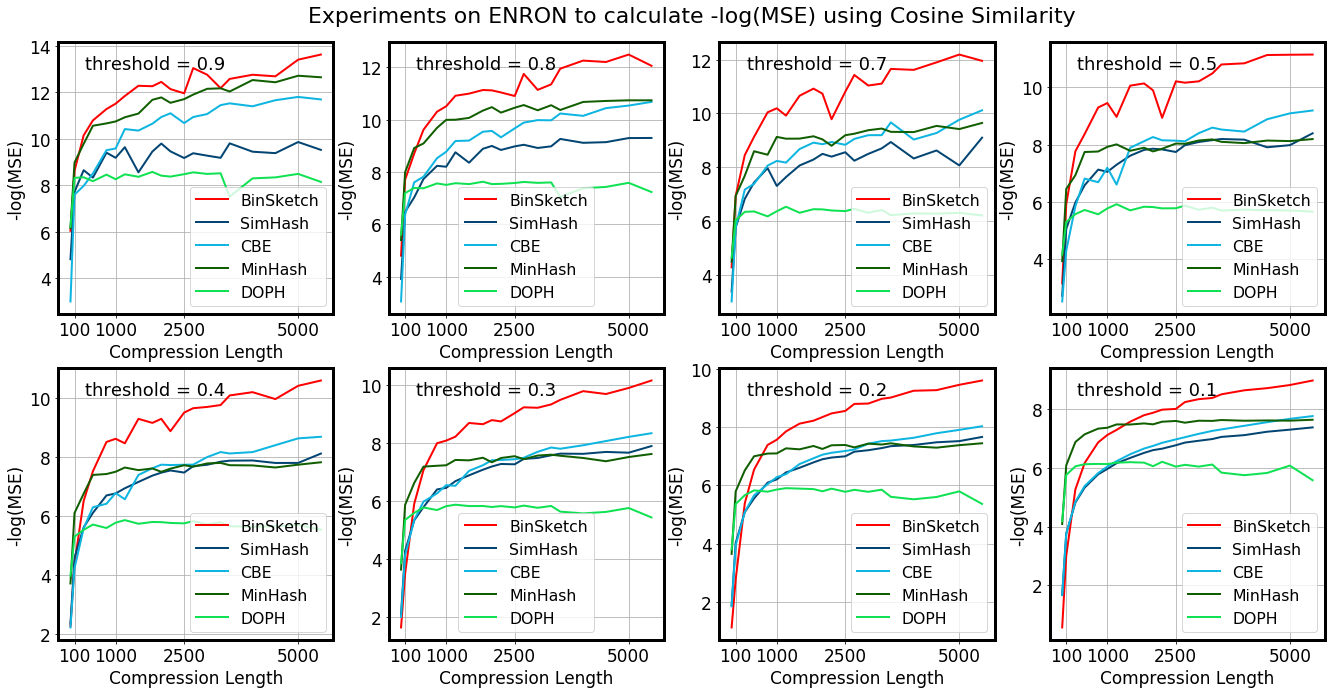}
\includegraphics[height=7.4cm,width=16.2cm]{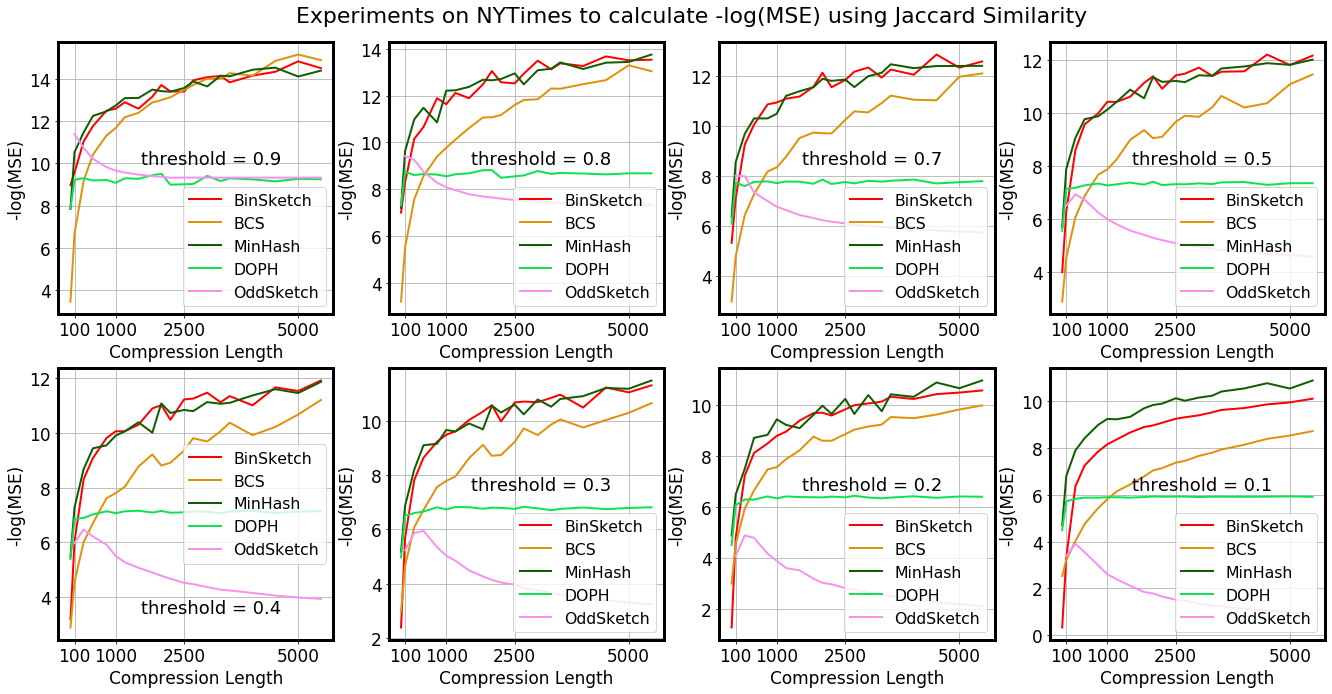}
\includegraphics[height=7.4cm,width=16.2cm]{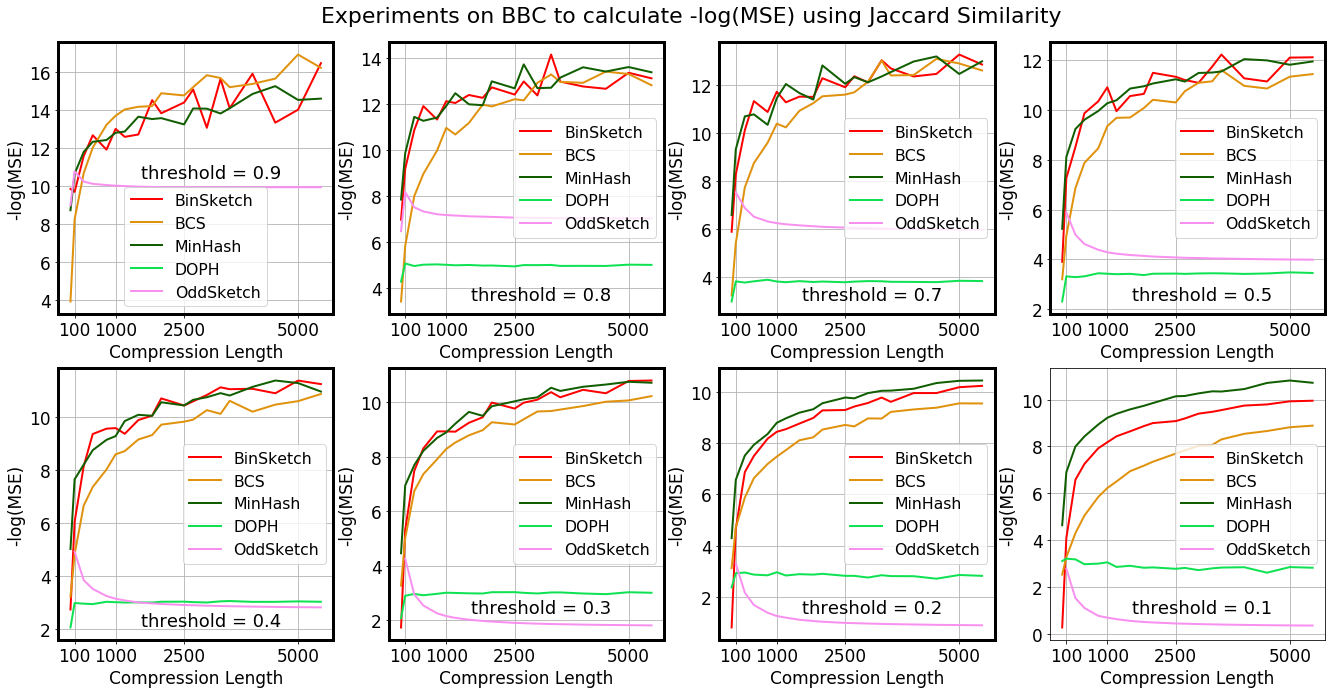}

\caption{\footnotesize{Comparison of $-\log(\MSE)$ measure on  Enron,  NYTimes, and BBC datasets. A higher value is an indication of better performance.}}
\label{fig:neglogMSE}
\end{figure*}

\begin{figure*}
\centering
\includegraphics[height=3.8cm,width=16.5cm]{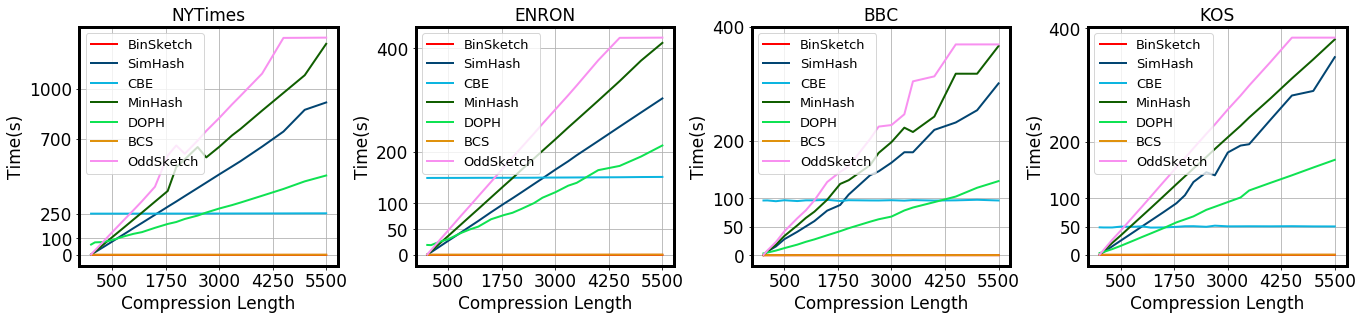}
\caption{\footnotesize{Comparison of compression times on NYTimes, ENRON, KOS and BBC datasets.}}
\label{fig:compressiontime}
\end{figure*}

\begin{figure*}
\centering
\includegraphics[height=8cm,width=16.5cm]{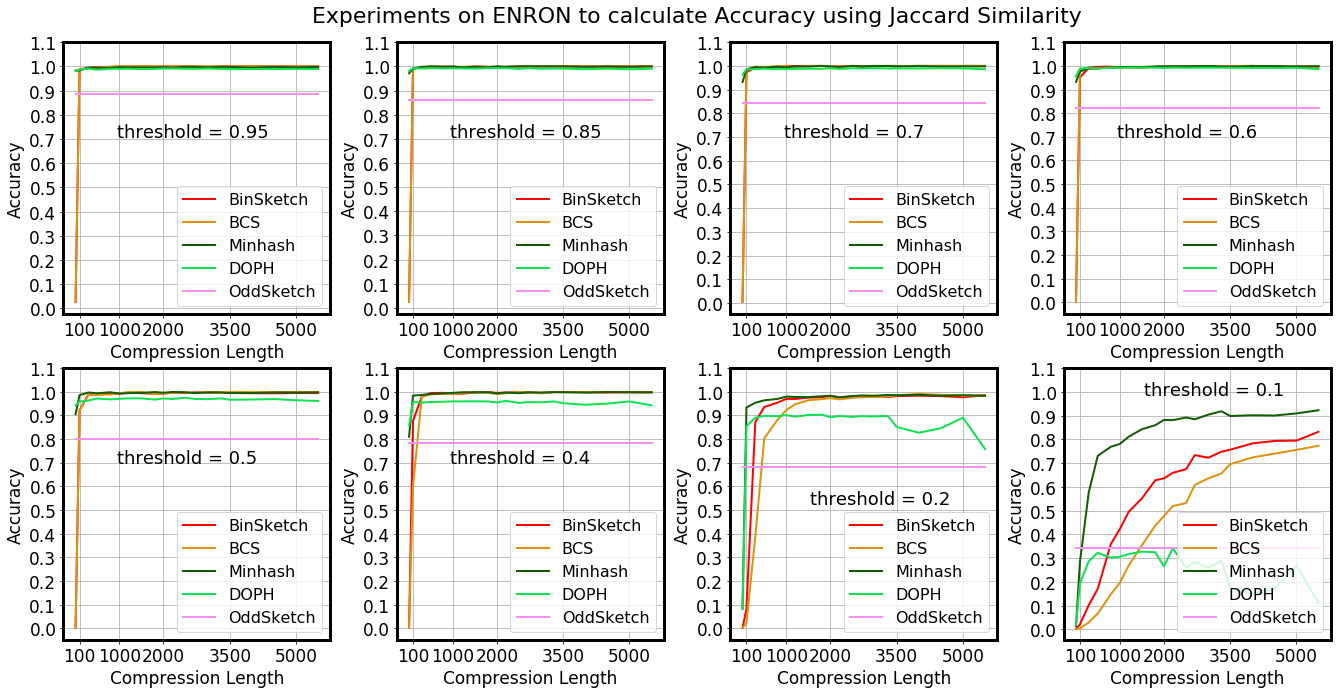}
\includegraphics[height=8cm,width=16.5cm]{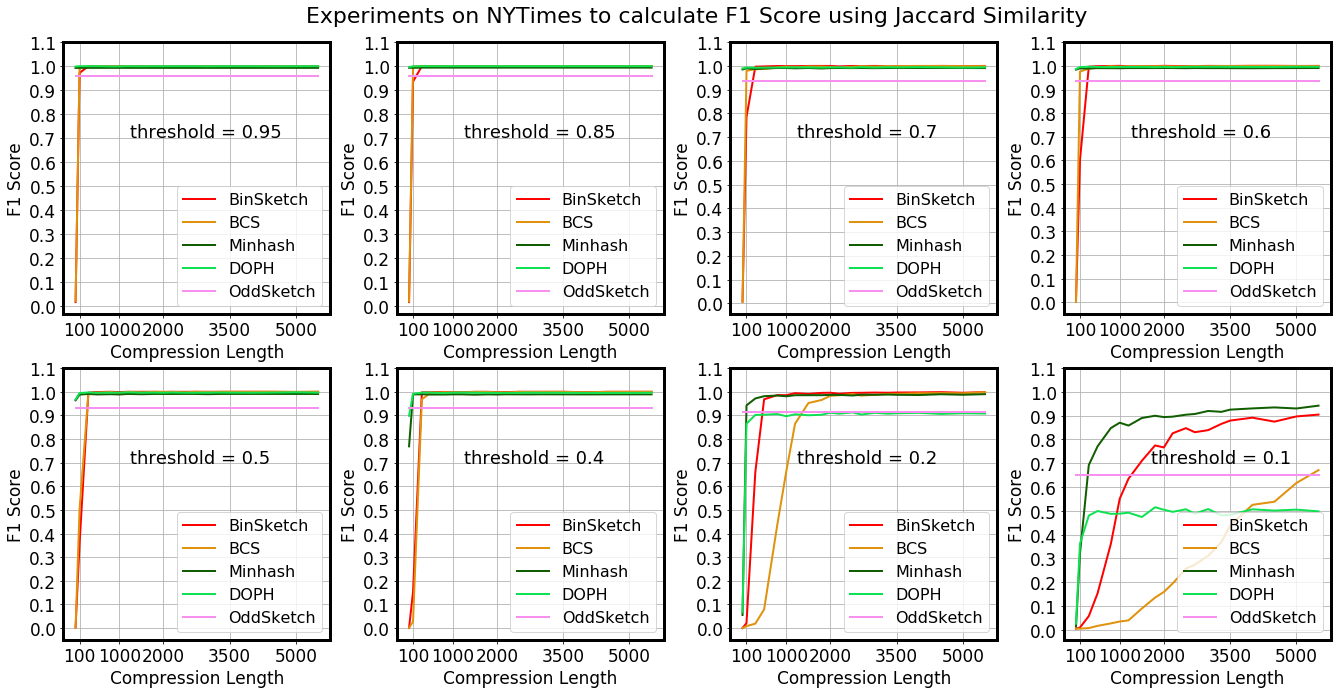}
\includegraphics[height=8cm,width=16.5cm]{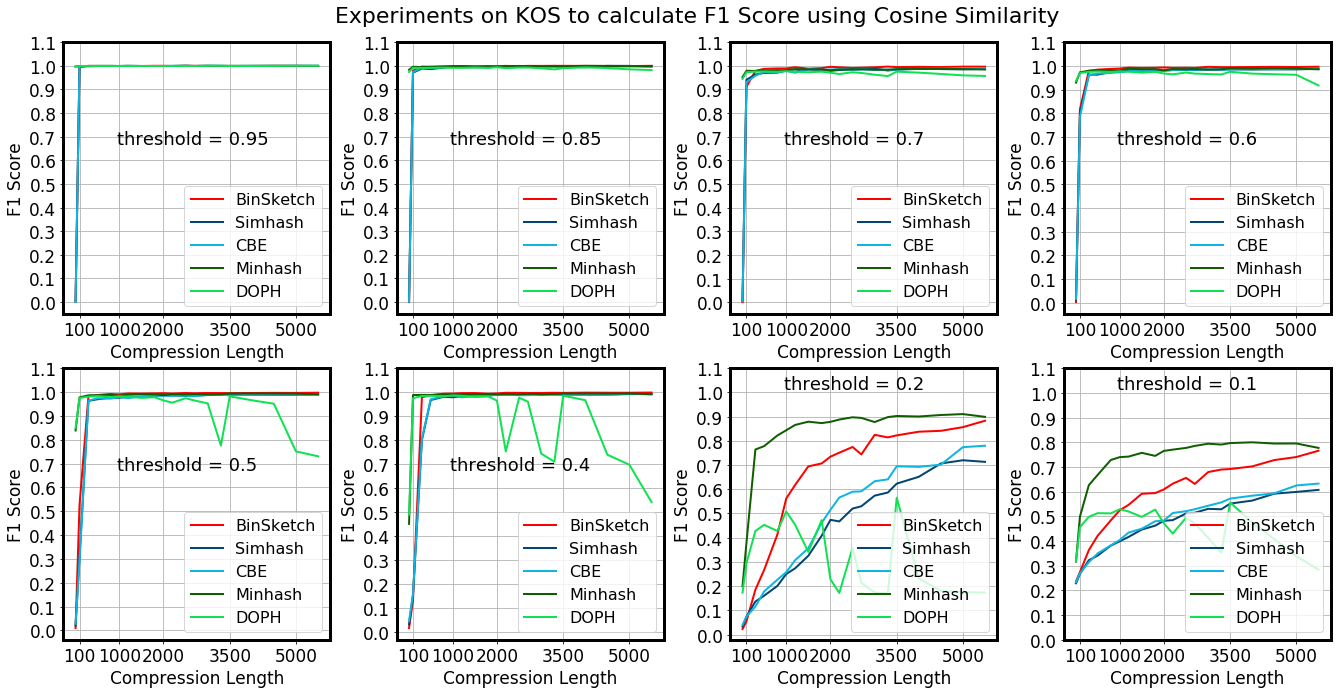}
\caption{\footnotesize{Comparison of Accuracy and $\F$ score measures on ENRON, NYTimes and KOS datasets.}}
\label{fig:Accuracy_main}
\end{figure*}
 
\paragraph*{Insights} 
 
We summarize Accuracy and $\F$ score results in Figure~\ref{fig:Accuracy_main}.
For Jaccard Similarity, on both Accuracy and $\F$ score measure, $\binsketch$ significantly outperformed $\BCS$,  $\DOPH$, and $\odd$ while 
its performance was comparable  \textit{w.r.t.} $\minhash$.  For Cosine  similarity, on higher and intermediate threshold values, 
$\binsketch$ outperformed  all the other candidate algorithms. 
However, on the lower threshold values, $\minhash$ offered the most accurate sketch followed by $\binsketch$.

\paragraph{Efficiency of $\binsketch$}  
We comment on the efficiency of $\binsketch$ with the other competing algorithms
and summarize our results in Figure~\ref{fig:compressiontime}.  
We noted the time required to compress the original dataset using all the competing algorithms. For a given compression length, the compression time of $\odd$ varies based on the similarity threshold. Therefore, we consider taking their average.  
We notice that the time required by $\binsketch$ and $\BCS$ is negligible for all values of $\N$ and on all the datasets. 
Compression time of $\CBE$ is very higher than ours, however, it is independent of the compression length $\N$.  
After excluding some initial compression lengths, the compression time of $\odd$ is the highest, and grows linearly with $\N$, as it requires running $\minhash$ on the original dataset. 
For the remaining algorithms, their respective compression time grows linearly with $\N$. 

\section{Summary and open questions}
In this work, we proposed a simple dimensionality reduction algorithm -- $\binsketch$ -- for sparse binary data. 
 $\binsketch$  offer an efficient dimensionality reduction/sketching algorithm, 
which compresses a given $d$-dimensional binary dataset to a relatively smaller $\N$-dimensional binary sketch, 
while simultaneously maintaining estimates for multiple similarity measures such as Jaccard Similarity, Cosine Similarity, 
Inner Product, and Hamming Distance, on the same sketch. 
The performance of $\binsketch$ was significantly better than $\BCS$~\cite{KulkarniP16,JS_BCS} while the compression (dimensionality reduction) time of these two algorithms were somewhat very comparable. 
$\binsketch$ obtained a significant speedup in  compression time \textit{w.r.t} other candidate algorithms ($\minhash$~\cite{BroderCFM98,Shrivastava014}, $\simhash$~\cite{simhash}, $\DOPH$~\cite{DOPH}, $\CBE$~\cite{CBE}) while it simultaneously offered a comparable performance guarantee.

 We want to highlight the error bound presented in Theorem~\ref{thm:IP_est} is due to a worst-case analysis, 
 which potentially can be  tightened. We state this as an open question of the paper. 
 Our experiments on 
 real datasets establish this. 
 For example, for the inner product (see Figure~\ref{fig:MSE_IP_NYTimes}), 
 we show that the Mean Square Error $(\MSE)$ is almost zero even for compressed dimensions that are much lesser than the bounds stated
 in the Theorem. 
Another important open question  is to  derive a lower
bound on the size of a sketch that is required to efficiently and accurately
derive similarity values from compressed sketches. 
 Given the simplicity of our method,
we hope that it will get adopted in practice. 

\balance
\bibliographystyle{plain}
\bibliography{reference} 
\clearpage
\newpage
\balance
\appendices
 \section{Proof of Observation~\ref{obs:1}}\label{appendix:obs_1}
In this section we prove that $n^{n_a} \ge 1/2$.
For this first we derive an upper bound of $\frac{1}{2}$ on $n_{a_s}/N$.

Let $P$ denote the expression $\sqrt{\frac{\tb}{2}\ln \frac{2}{\delta}}$
appearing in Lemma~\ref{lemma:1}. Using this lemma, $n_{a_s} \le \E(|a_s|) + P$.
Observe that $\E(|a_s|) = N(1-n^{|a|}) \le N(1-n^{\psi})$ since $|a| \le \psi$
and $n \in (0,1)$. Furthermore, since $n^{\psi} = (1-\frac{1}{N})^{\psi} \ge 1 -
\frac{\psi}{N} \ge \frac{1}{2}$, we get the upper bound $n_{a_s}/N \le \frac{1}{N} \left( N
\frac{\psi}{N} + P \right) = \frac{\psi}{N} + \frac{P}{N} =
\frac{1}{\sqrt{\frac{\psi}{2} \ln \frac{2}{\delta}}} + \frac{1}{\psi}$. For
reasonable values of $\delta$ and $\psi$, both $\psi$ and $\sqrt{\frac{\psi}{2}
\ln \frac{2}{\delta}}$ are at least 4; thus, we get the bound of $n_{a_s}/N \le
\frac{1}{2}$ and this leads us to the bound $n^{n_a} = 1 - n_{a_s}/N  \ge
\frac{1}{2}$.

\section{Proof of Lemma~\ref{lemma:IP_estimation}}

\label{appendix:proof-thm-ip-est}
In this section we derive an upper bound on
\begin{align*}
    B = & \Big| |a|-n_a ~+~ |b| - n_b ~+~ \\
    & \tfrac{1}{\ln \tfrac{1}{n}} \ln\left[n^{|a|} + n^{|b|} +
    \frac{\E[\IPS{a_s,b_s}]}{\N}-1\right] ~-~ \\
    & \tfrac{1}{\ln \tfrac{1}{n}}\ln \left[ n^{n_a} +
    n^{n_b} + \frac{n_{a_s,b_s}}{\N} - 1 \right] \Big|
\end{align*}

\begin{proof}[Proof]
    We first apply triangle inequality and Lemma~\ref{lemma:2} to obtain
    $$	B \le \frac{4/\psi}{\ln \tfrac{1}{n}} + \frac{4/\psi}{\ln \tfrac{1}{n}}
    + \frac{1}{\ln \tfrac{1}{n}} \left| \ln \frac{n^{n_a} +
	n^{n_b} + \frac{n_{a_s,b_s}}{\N} - 1}{n^{|a|} + n^{|b|} +
	\frac{\E[\IPS{a_s,b_s}]}{\N}-1} \right| $$

    Next we derive an upper bound for the last term for which we require the
    next few observations. Let $U$ denote $n^{n_a} +
	n^{n_b} + \frac{n_{a_s,b_s}}{\N} - 1$, $V$ denote $n^{|a|} + n^{|b|} +
	\frac{\E[\IPS{a_s,b_s}]}{\N}-1$, and $W$ denote $|\ln \frac{U}{V}|$.

    \begin{obs}By expanding $n=(1-\frac{1}{N})$ and employing $(1+x)^r \ge 1+rx$, we obtain that
	$n^{n_a} +
	n^{n_b} + \frac{n_{a_s,b_s}}{\N} - 1 \ge 1 - \frac{|a_s| + |b_s| +
	\IPS{a_s,b_s}}{N} > 0$ since $\IPS{a_s,b_s} \le |a_s|$.
    \end{obs}
    
    \begin{obs}Using Lemma~\ref{lemma:E_ip_ab}, $n^{|a|} + n^{|b|} +
	\frac{\E[\IPS{a_s,b_s}]}{\N}-1 = n^{|a| + |b| + \IPS{a,b}} > 0$ for
	non-zero $a$ and $b$.
    \end{obs}

    These observations ensure that the terms inside the logarithm are indeed
    positive.

    Next we upper bound $W$ by employing the inequality $\left| \ln \frac{A}{B} \right| \le
    \frac{|A-B|}{\max(A,B)}$ that holds for non-negative $A,B$ and can be
    derived from the standard inequality $\ln x \le x - 1$ for $x > 0$. Here,
    set $A=U$ and $B=V$. Then, using triangle inequality
    \begin{align*}
	|U - V| & \le |n^{n_a} - n^{|a|}| + |n^{n_b} - n^{|b|}| +
	\frac{|\E[\IPS{a_s,b_s}] - n_{a_s,b_s}|}{N} \\
	& \le 3 \frac{1}{\psi} \text{( using Lemma~\ref{lemma:3} and the next
	observation)}
    \end{align*}

    \begin{obs}
	These claims appear in the proof of Lemma~\ref{lemma:2}: $|n^{|a|} -
	n^{n_a}| < \frac{1}{\psi}$ and $n^{n_a} \ge n^{|a|} - \frac{1}{\psi}$. 
	Similarly, $|n^{|b|} - n^{n_b}| < \frac{1}{\psi}$ and $n^{n_b} \ge
	n^{|b|} - \frac{1}{\psi}$.
    \end{obs}

    We need one final observation to compute $\max(U,V)$.
    \begin{obs}
	Using Lemma~\ref{lemma:3}, $\frac{n_{a_s,b_s}}{N} \ge
	\frac{\E(\IPS{a_s,b_s})}{N} - \frac{1}{\psi}$.
    \end{obs}

    Based on the last two observations we can compute
    \begin{align*}
	U = & n^{n_a} + n^{n_b} + \frac{n_{a_s,b_s}}{N} - 1 \\
    \ge & n^{|a|} + n^{|b|} +
    \frac{\E(\IPS{a_s,b_s})}{N} - \frac{3}{\psi} - 1 \\
	= & V - \frac{3}{\psi} = n^{|a|+|b|+\IPS{a,b}} - \frac{3}{\psi}
    \end{align*}
    Therefore, if $U \ge V$, then $\max(U,V) = U \ge V - \frac{3}{\psi}$ and if $V >
    U$, then $\max(U,V) = V \ge V - \frac{3}{\psi}$. This leads to:
    \begin{align*}
	\max(U,V) & \ge V - \frac{3}{\psi} =
	(1-\frac{1}{N})^{|a|+|b|+\IPS{a,b}} - \frac{3}{\psi} \\
	& \ge 1 - \frac{|a|+|b|+\IPS{a,b}}{N} - \frac{3}{\psi} \ge 1 -
	\frac{3\psi}{N} - \frac{3}{\psi}\\
	& \ge 1 - \frac{3}{\sqrt{\frac{\psi}{2} \ln \frac{2}{\delta}}} -
	\frac{3}{\psi}
    \end{align*}
    which is at least $\frac{1}{2}$ for reasonable values of $\psi$ and
$\delta$.

    Now we gather all the upper bounds of the expressions appearing in $B$ and
compute the upper bound as stated in the lemma.

    $$B \le \frac{4/\psi}{\ln \frac{1}{n}} + \frac{4/\psi}{\ln \frac{1}{n}} +
\frac{1}{\ln \frac{1}{n}}\frac{3/\psi}{1/2} = \frac{14/\psi}{\ln \frac{1}{n}}
\le \frac{14N}{\psi} = 14\sqrt{\frac{\psi}{2}\ln \frac{2}{\delta}}$$

    Of course this bound holds when the upper bounds on $(|a|-n_a)$,
    $(|b|-n_b)$ and $(\E[\IPS{a_s,b_s}] - n_{a_s,b_s})$ are correct and
    each of them is incorrect with probability at most $\delta$. Therefore,
    using Union-bound, we can say that our upper bound as required in the lemma
    can be incorrect with probability at most $3\delta$.
\end{proof}

\end{document}